\documentclass{article}
\usepackage{ismir,cite}
\usepackage{aliascnt}
\usepackage{tabu}
\usepackage{amssymb,amsmath,mathtools,amsthm}
\usepackage{booktabs}
\usepackage{multirow}
\usepackage[english]{babel}
\newtheorem{proposition}{Proposition}
\usepackage{epstopdf}
\usepackage{gensymb}
\usepackage{siunitx}
\newcommand{\W}{\mathbf{W}}
\newcommand{\D}{\mathbf{D}}

\usepackage[draft]{todonotes}   

\usepackage{times}
\usepackage{url}
\usepackage{color}
\usepackage{multirow}
\usepackage{multicol}
\usepackage{graphicx}
\usepackage{caption}
\usepackage{tikz}
\usetikzlibrary{calc,tikzmark}
\usepackage[shortcuts]{extdash}
\usepackage{blindtext}
\usepackage[ruled,norelsize,linesnumbered]{algorithm2e}

\newcommand{\comment}[1]{\textcolor{blue}{}}
\newcommand{\mg}[1]{\textcolor{red}{}}

\makeatletter
\def\inmod#1{\allowbreak\mkern5mu{\operator@font mod}\,\,#1}
\makeatother

\urldef{\discovery}\url{http://www.music-ir.org/mirex/wiki/2017:Discovery_of_Repeated_Themes_%26_Sections}

\title{Learning Complex Basis Functions for Invariant Representations of Audio}
\multauthor
{Stefan Lattner,$^{1}$ Monika D{\"o}rfler,$^2$ Andreas Arzt$^{3}$\vspace{1mm} }
{$^1$ Sony Computer Science Laboratories (CSL), Paris, France\\
$^2$ NuHAG, Faculty of Mathematics, University Vienna\\
$^3$ Institute of Computational Perception, JKU Linz\\
}
\def\authorname{Stefan Lattner, Monika D\"orfler, Andreas Arzt}

\begin{document}
\sloppy
\maketitle

\begin{abstract}
Learning features from data has shown to be more successful than using hand-crafted features for many machine learning tasks. In music information retrieval (MIR), features learned from windowed spectrograms are highly variant to transformations like transposition or time-shift. Such variances are undesirable when they are irrelevant for the respective MIR task. We propose an architecture called Complex Autoencoder (CAE) which learns features invariant to orthogonal transformations. Mapping signals onto complex basis functions learned by the CAE results in a transformation-invariant ``magnitude space'' and a transformation-variant ``phase space''. The phase space is useful to infer transformations between data pairs. When exploiting the invariance-property of the magnitude space, we achieve state-of-the-art results in audio-to-score alignment and repeated section discovery for audio.
A PyTorch implementation of the CAE, including the repeated section discovery method, is available online.\footnote{\url{https://github.com/SonyCSLParis/cae-invar}}
\end{abstract}

\section{Introduction}\label{sec:introduction}
Learning from audio data most commonly involves some prior processing of the raw sound signals.
The most popular features are derived from a spectrogram, which consists of the magnitude values of the Fourier transform of a windowed signal of interest. 
In a Fourier transform, a signal is projected onto sine and cosine functions of different frequencies.
One of the main reasons for the spectrogram to be more useful than the usage of the Fourier coefficients in their complex form is the fact that the magnitude spectrum of a signal is invariant to a translation of the original signal. 

This invariance to translation, desirable for most learning problems in audio,  results from the fact that cosine and sine represent  the real and imaginary parts,  respectively, of the \emph{complex eigenvectors} of  translation.
More generally, the eigenvectors of an orthogonal transformation (e.g., translation, rotation, reflection, but most general all permutations - ``shuffling pixels'') constitute an orthonormal basis of complex vectors with corresponding eigenvalues of magnitude $1$. Hence, as we shall see in detail,  the absolute value of a signal's coefficients with respect to this basis is invariant to that transformation.
We harness this invariance property for learning representations invariant to different orthogonal transformations.

In particular, transposition-invariance is an essential property for several MIR tasks, including alignment tasks, repeated section discovery, classification tasks, cover song detection, query by humming, or representations of acapella recordings with pitch drift.
Different methods have aimed at learning transposition-invariant representations.
For example, in \cite{walters2012intervalgram} close time steps in chromagrams are cross-correlated in order to calculate distances between pitch classes, and in \cite{lattner2018learning}, successive n-grams of constant-Q transformed (CQT) representations of audio are compared using a Gated Autoencoder (GAE) architecture. 
Most similar to our approach, the transposition-invariant magnitudes of Fourier transformations applied to chromagram-like representations of audio are facilitated in \cite{DBLP:conf/ismir/Bertin-MahieuxE12} and \cite{marolt2008mid}.
However, instead of using 2D Fourier transforms with fixed basis functions, we learn the relevant basis functions starting from CQT representations of audio.
This learning of basis functions has some advantages over using pre-defined bases.
For example, the number of basis elements necessary to discriminate between signals can be reduced compared to a common Fourier transform (e.g., for transposition- and time-shift invariance we use $M=256$ basis functions for input dimensionality $N=3840$, while usually $N=M$).
Furthermore, our approach is generic and has the potential to learn other musically interesting invariances (e.g., towards tempo-change, diatonic transposition, inversion, or retrograde).

The contribution of this paper is a simple training method for learning invariant representations from data pairs and its application to two MIR tasks.
First, we show that when using the features learned by the Complex Autoencoder (CAE) from audio in CQT representation, we can improve the state-of-the-art in a transposition-invariant repeated section discovery task in audio.
Second, the CAE features prove useful in an audio-to-score alignment task, where we show that most of the time, they yield better results than Chroma features and features calculated with a GAE.
We also compare the CAE with a GAE in classifying rotated MNIST digits, based on rotation-invariant features learned by the CAE.
The reason we also perform experiments on MNIST is that it allows us to show the efficacy of the model with respect to rotation-invariance.
Furthermore, the class labels available in the MNIST dataset help to highlight the different clusters in the rotation-invariant space (see Figure \ref{fig:pca}).

In particular for translations, the model can be interpreted as measuring distances in the data.
Training the CAE for transposition and time-shift invariance on short windows of audio in CQT representation, therefore, leads to representations of rhythmic structures and tonal relationships present in the windows, what we exploit in the repeated section discovery task.
Representing rhythmic structures is less critical in music alignment tasks; it can even be disadvantageous when the aligned signals differ in tempo.
We show in the alignment task that it is sufficient to train the CAE only for transposition-invariance (i.e., time-shift transformation) on rather short n-grams of audio in CQT representation.
This is because compared to repeated section discovery, where rhythmic patterns can help to identify similar parts, in the alignment task, a dynamic time-warping algorithm keeps track of the respective positions in the music pieces.

The CAE can be trained in an unsupervised manner on data pairs obeying the relevant transformations.
Thereby, we obtain a ``magnitude space'' and a ``phase space'', as it is known from a Fourier transform.
The ``magnitude space'' of the CAE is invariant to all the learned transformations.
Remarkably, the \emph{phase shifts} a projected signal undergoes during a transformation (i.e., the relative vector in the ``phase space'' of the CAE) are discriminative with respect to the type and the distance of a transformation.
This is an interesting property which could be exploited for determining types of relations between musical fragments in structure analysis tasks.

The paper is structured as follows. In Section \ref{sec:rel-work} existing work related to the proposed method is discussed. In Section \ref{sec:model-maths} we describe the model and its mathematical background and Section \ref{sec:training} describes the general training procedure. In Section \ref{sec:experiments}, we show results on three different tasks: discovery of repeated themes and sections, audio-to-score alignment, and classification of MNIST digits. We end the paper with a conclusion and a discussion of possible directions for future work (Section \ref{sec:conclusion}).

%


\section{Related work}\label{sec:rel-work}
Generally, mid-level representations in neural networks are highly variant to transformations in the input.
The most common and well-known way to obtain shift-invariance in convolutional architectures is max-pooling \cite{DBLP:conf/icml/BoureauPL10}.
However, full shift-invariance can only be achieved step-wise by applying max-pooling over several layers.
A whole line of research therefore aims to obtain representations invariant to different kinds of transformations using other approaches.
Inspiration for the proposed model was drawn from \cite{memisevic2013aperture}, where complex basis functions are learned using a GAE.
An approach similar to ours is to facilitate harmonic functions or wavelets, either in weight initialization \cite{DBLP:conf/icdar/ZhongJX15, DBLP:conf/iccv/OuyangW13}, for modulating learned filters using Gabor functions \cite{GaborConvNetworks}, or for using fixed wavelets in scattering transforms \cite{DBLP:conf/cvpr/SifreM13, DBLP:journals/pami/BrunaM13}.
Similarly, harmonic functions can be pre-defined, e.g., to obtain rotation invariance in convolutional architectures \cite{DBLP:conf/cvpr/WorrallGTB17}, or learned, e.g., by assuming ``temporal slowness'' of features in videos \cite{DBLP:conf/cvpr/LeZYN11, olshausen2007bilinear}, while pitch-invariant  timbral  features are learned in \cite{Pons18} by enabling convolution  through the  frequency domain.

Most of the approaches mentioned so far (including our approach) aim at invariances to relatively simple, affine transformations.
Invariances to more complex, non-linear transformations are usually achieved by redundancy (e.g., an object is presented from different camera angles or under different lighting conditions), which typically requires bigger architectures.
That way, invariance can be learned by an explicit transformation of the input \cite{jaderberg2015spatial}, by enforcing similarity in the latent space \cite{DBLP:conf/iros/MatsuoFS17}, or by using a Siamese architecture and pre-defined transformation sets \cite{TiPooling}.
Other methods involve rotating convolution kernels during training \cite{DBLP:conf/cvpr/ZhouYQJ17} and dealing with input deformations using learned, dynamic convolution grids \cite{DBLP:conf/iccv/DaiQXLZHW17}. In  \cite{Diele14} an end-to-end CNN which acts on raw audio learns Gabor-like filters similar to those extracted by the CAE, see Figure~\ref{fig:basis_vectors}.

\section{Model and mathematical background}\label{sec:model-maths}
\begin{figure}
\centering
\includegraphics[width=.8\linewidth]{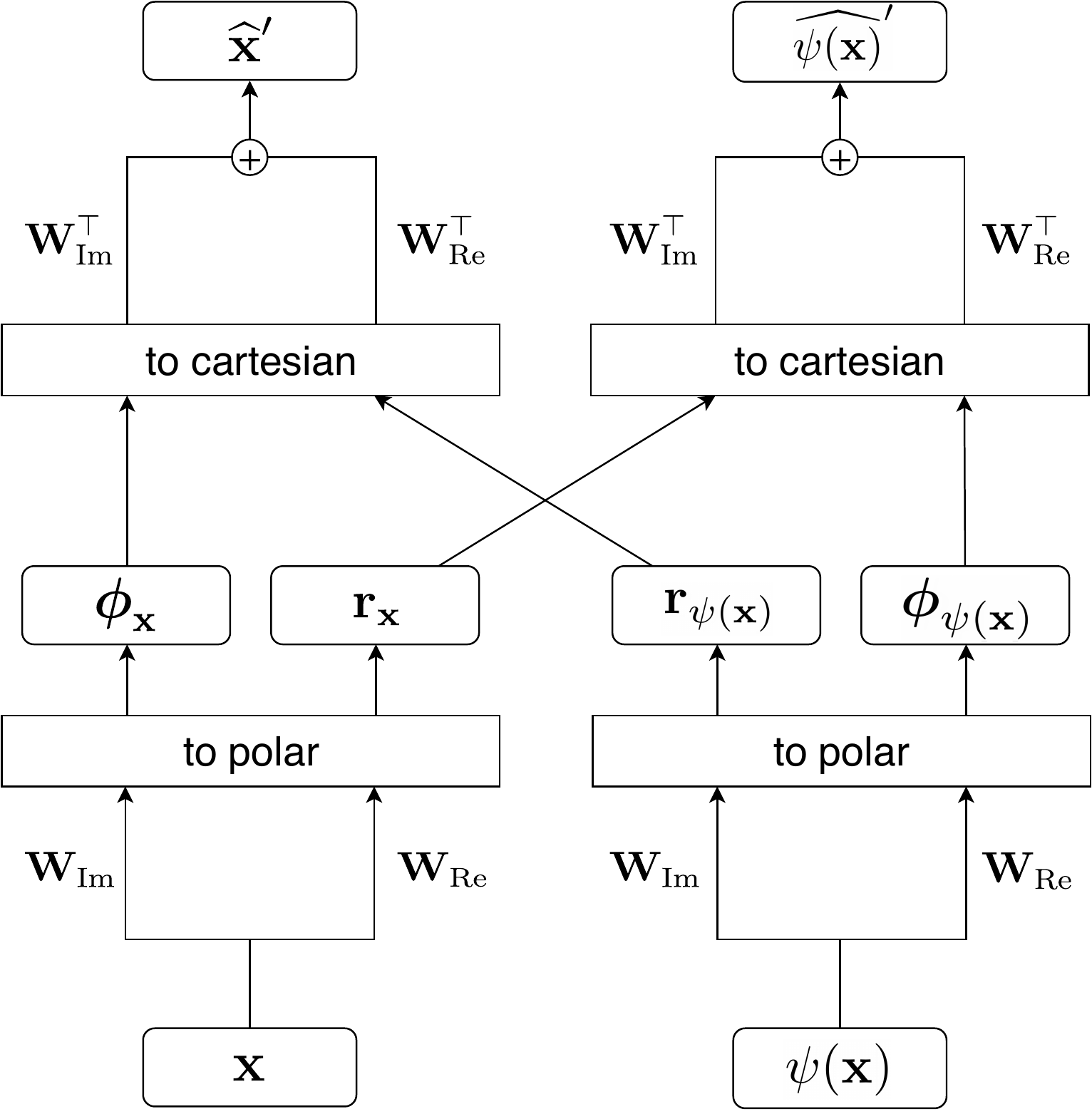} \\
\caption{Schematic illustration of reconstruction during training. Both the input $\mathbf{x}$ and its transformed counterpart $\psi(\mathbf{x})$ are projected onto complex basis pairs $\{\mathbf{W}_{\text{Re}}, \mathbf{W}_{\text{Im}}\}$ and expressed in polar form. Then, using the swapped magnitude vectors $\{\mathbf{r}_{\psi (\mathbf{x})}, \mathbf{r}_{\mathbf{x}}\}$ and the original phase vectors $\{\boldsymbol{\phi}_{\mathbf{x}}, \boldsymbol{\phi}_{\psi (\mathbf{x})}\}$, the data is reconstructed by performing the inverse operations.}
\label{fig:model}
\vspace{-4mm}
\end{figure}

\begin{figure*}
\centering
\includegraphics[width=0.9\linewidth]{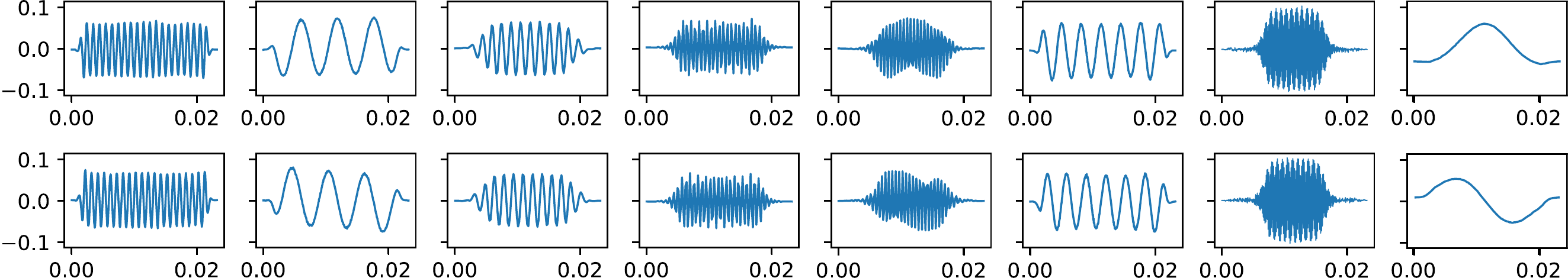} \\
\caption{Some examples of real (top) and imaginary (bottom) basis vectors learned from audio signals (time in seconds).}
\label{fig:basis_vectors}
\end{figure*}

We aim at learning orthogonal transformations  encoding certain invariances of a class of signals which are known or assumed to be useful for a particular learning task at hand. To this end, we leverage the particular properties of orthogonal transformations, which we now describe.
A transformation $\psi : \mathbb{R}^N \rightarrow \mathbb{R}^N$ is orthogonal if 
$\langle \psi (\mathbf{x}), \psi (\mathbf{y})\rangle = \langle \mathbf{x}, \mathbf{y} \rangle$ for all $\mathbf{x},\mathbf{y}\in\mathbb{R}^N$. By $\langle \mathbf{x},\mathbf{y}\rangle$  we denote the inner product on $\mathbb{R}^N$ or $\mathbb{C}^N$, respectively.
Orthogonal transformations are distinguished by the fact that they possess a diagonalization
with eigenvalues which all have absolute value $1$. Hence, in any non-trivial case, the eigenvalues are complex and so are the corresponding eigenvectors.  More precisely, if $\psi$ is orthogonal, there exists a unitary
matrix\footnote{A complex matrix  $\W$ is unitary, if $\W^\ast = \W^ {-1}$.} $\W$ and eigenvalues $\lambda_j$, $j = 1, \ldots , N$, with $|\lambda_j | = 1$
for all $j$, such that 
\begin{equation}\label{Eq:orth}
\psi (\mathbf{x}) = \W^\ast \D\W \mathbf{x}
\end{equation}
Here $\D$ denotes the diagonal $N\times N$ matrix with the eigenvalues $\lambda_j$ in the diagonal. We hence have the following statement. 
\begin{proposition}\label{Prop1}
If an orthogonal transformation  $\psi : \mathbb{R}^N \rightarrow \mathbb{R}^N$ is diagonalised  by a unitary matrix $\W$, then the feature vector given by $|\W \mathbf{x}| $ for all $\mathbf{x}\in  \mathbb{R}^N$ is invariant to $\psi$. In other words, we have $|\W \mathbf{x} |  = |\W \psi (\mathbf{x})| $ for all $\mathbf{x}\in  \mathbb{R}^N$.
\end{proposition}
\begin{proof}
According to \eqref{Eq:orth} and since $\W \W^\ast \mathbf{x}= \mathbf{x}$, we have 
\begin{equation}
\W\psi (\mathbf{x}) = \D\W \mathbf{x}, 
\end{equation}
which can be written coordinate-wise as  
$$\langle w_j, \psi (\mathbf{x})\rangle = \lambda_j \langle w_j, \mathbf{x} \rangle,\, j = 1,\ldots, N, $$ where $w_j$ denotes the $j-$th row of the complex, unitary matrix $\W$. Hence, since
$|\lambda_j | = 1$
  for all $j$, we have $|\langle w_j, \psi (\mathbf{x})\rangle| = |\langle w_j, \mathbf{x} \rangle |$  for all $\mathbf{x}\in\mathbb{R}^N$ and thus 
  $|\W \mathbf{x}|  = |\W \psi (\mathbf{x})| $  as claimed. 
\end{proof}

In CAE-learning, we can only deal with real weights and are usually interested in learning less than $N$ basis vectors. Hence, we split an $M\times N$-dimensional submatrix of the unitary, complex matrix of eigenvectors $\W$ into a real and an imaginary part $\mathbf{W}_{\text{Re}} \in \mathbb{R}^{M\times N}$ and $\mathbf{W}_{\text{Im}} \in \mathbb{R}^{M\times N}$ which map $\mathbf{x}$ onto the real and imaginary part of $\W \mathbf{x}$, respectively.
The complex data  $\mathbf{W}_{\text{Re}} \mathbf{x} + i\mathbf{W}_{\text{Im}} \mathbf{x}$ is then expressed in its polar form by a phase vector $\boldsymbol{\phi}_{\mathbf{x}} \in [0, 2\pi)^M$ 
and a magnitude vector $\mathbf{r}_{\mathbf{x}} \in \mathbb{R}_{\geq 0}^M$.

According to Proposition~\ref{Prop1}, assuming that $\mathbf{W}_{\text{Re}}+ i\mathbf{W}_{\text{Im}}$ consist of orthonormal eigenvectors of $\psi$ leads to  the magnitude of projections of a signal $\mathbf{x}$ and a transformed version $\psi(\mathbf{x})$ onto these eigenvectors  to be equal.
This property is imposed during training by expressing $\W \mathbf{x}$ and $\W \psi (\mathbf{x})$ in their 
respective polar forms 
\begin{equation}
\boldsymbol{\phi}_{\mathbf{x}} = \text{atan2}(\mathbf{W}_{\text{Re}} \mathbf{x}, \mathbf{W}_{\text{Im}} \mathbf{x}),
\end{equation}
and
\begin{equation}
\mathbf{r}_{\mathbf{x}} = \sqrt{(\mathbf{W}_{\text{Re}} \mathbf{x})^2 + (\mathbf{W}_{\text{Im}} \mathbf{x})^2},
\end{equation}
and swapping the magnitude vectors before reconstruction.

Accordingly, we reconstruct $\mathbf{x}$ given its own phase representation $\boldsymbol{\phi}_{\mathbf{x}}$ and the magnitude representation of the transformed signal $\mathbf{r}_{\psi (\mathbf{x})}$ as follows (see Figure \ref{fig:model}):
\begin{equation}
\mathbf{\widehat{x}}' = \mathbf{W}_{\text{Re}}^\top (\mathbf{r}_{\psi (\mathbf{x})} \cdot \sin{\boldsymbol{\phi}_{\mathbf{x}}}) + \mathbf{W}_{\text{Im}}^\top (\mathbf{r}_{\psi (\mathbf{x})} \cdot \cos{\boldsymbol{\phi}_{\mathbf{x}}}). 
\end{equation}
Likewise, we reconstruct the transformed signal $\psi (\mathbf{x})$ as
\begin{equation}
\widehat{\psi (\mathbf{x})}' = \mathbf{W}_{\text{Re}}^\top (\mathbf{r}_{\mathbf{x}} \cdot \sin{\boldsymbol{\phi}_{\psi (\mathbf{x})}}) + \mathbf{W}_{\text{Im}}^\top (\mathbf{r}_{\mathbf{x}} \cdot \cos{\boldsymbol{\phi}_{\psi (\mathbf{x})}}).
\end{equation}

\begin{table*}
\centering
\footnotesize
\begin{tabular}{llllllllllllll}
\toprule
Algorithm     & $F_{\text{est}}$  & $P_{\text{est}}$  & $R_{\text{est}}$  & $F_{\text{o(.5)}}$   & $P_{\text{o(.5)}}$   & $R_{\text{o(.5)}}$
 & $F_{\text{o(.75)}}$   & $P_{\text{o(.75)}}$   & $R_{\text{o(.75)}}$  & $\mathbf{F_3}$    & $P_3$    & $R_3$    & Time (s) \\
\midrule
\rule{-2pt}{2ex}
CA (ours) & 52.53 & 63.10 & 50.29 & 63.58 & 64.60 & 62.68 & 67.20 & 68.73 & 65.82 & \textbf{52.16} & \textbf{62.51} & 49.78  & \textbf{69} \\ 
GAE intervals \cite{lattner2018learning} & \textbf{57.67} & \textbf{67.46} & 59.52 & 58.85 & 61.89 & 56.54 & 68.44 & 72.62 & 64.86 & 51.61 & 59.60 & \textbf{55.13}  & 194 \\ 
VMO deadpan \cite{wang2015music} & 56.15 & 66.80  & 57.83 & \textbf{67.78} & \textbf{72.93} & \textbf{64.30}  & \textbf{70.58} & \textbf{72.81} & \textbf{68.66} & 50.60  & 61.36 & 52.25 & 96   \\
SIARCT-CFP \cite{collins2013siarct}    & 23.94 & 14.90  & \textbf{60.90}  & 56.87 & 62.90 & 51.90  & -     & -     & -     & -     & -     & -     & -    \\
Nieto  \cite{nieto2014identifying}      & 49.80  & 54.96 & 51.73 & 38.73 & 34.98 & 45.17 & 31.79 & 37.58 & 27.61 & 32.01 & 35.12 & 35.28 & 454 \\
\bottomrule
\end{tabular}
\caption{Different precision, recall and f-scores (adopted from \cite{wang2015music}, details on the metrics are given in \cite{collins2017discovery}) of different methods in the Discovery of Repeated Themes and Sections MIREX task, for symbolic music and audio. The $\mathbf{F_3}$ score constitutes a summarization of all metrics.}
\label{tab:secdiscovery}
\vspace{-2mm}
\end{table*}


The CAE is then trained by minimizing the symmetric reconstruction error %
%
%
\begin{equation}\label{eq:loss}
\frac{1}{N}\sum_i^N (x_i - \widehat{x}_i')^{p} + \frac{1}{N}\sum_j^N (\psi (x_j ) - \widehat{\psi (x_j )}' )^{p},
\end{equation}
where $p \in \{1,2\}$ has shown to work well in practice.
Training on sufficiently many transform pairs thus leads to learning the weights of the unitary matrix $\W$, which diagonalises $\psi$. While the magnitudes of the coefficient vectors $\W \psi (\mathbf{x})$ are equal to $\W \mathbf{x}$, the transformation itself is then represented by the differences in the phase vectors $\Delta\boldsymbol{\phi}=\boldsymbol{\phi}_{\mathbf{x}} - \boldsymbol{\phi}_{\psi (\mathbf{x})}$ (see Figure \ref{fig:pca}(b)).
As an example of complex basis vectors learned by the CAE, see Figure \ref{fig:basis_vectors}, where the CAE was trained on time-shifted audio signals in the time domain, yielding complex Gabor-like filters.
\section{Training}\label{sec:training}
For all the experiments described below, we choose 256 complex basis vectors and train the model for 500 epochs with a learning rate of 1e-3. We use a batch size of 1000, and we sample 100k transformations per epoch, generally picking random instances from the train set to be transformed.
The training data is standardized, and $50\%$ dropout is used on the input.
We set $p=1$ (see Equation \ref{eq:loss}) for the audio experiments, and $p=2$ for the MNIST experiment.
In the alignment experiment, we also penalize the mean of norms of all basis vectors and the deviation of the individual basis vectors' norms to the average norm over all basis vectors.
In the MNIST experiment, the norm of all basis vectors is set to $0.4$ after every batch.
For information about the training data see the respective experiment section below.

\vspace{-1mm}

\section{Experiments}\label{sec:experiments}
\subsection{Discovery of Repeated Themes and Sections}\label{sec:disc-repe-them}
\urldef\mirexurl\url{http://www.music-ir.org/mirex/wiki/2017:Discovery_of_Repeated_Themes_&_Sections}
In the MIREX task ``Discovery of Repeated Themes and Sections'',\footnote{\mirexurl{}} the performance of different algorithms to identify repeated (and possibly transposed) patterns in symbolic music and audio is tested.
The commonly used JKUPDD dataset \cite{collins2017discovery} contains 26 motifs, themes, and repeated sections annotated in 5 pieces by J. S. Bach, L.~v.~Beethoven, F. Chopin, O. Gibbons, and W. A. Mozart.
We use the audio versions of the dataset and preprocess them the same way as the training data described below.

The CAE is trained on 100 random piano pieces of the MAPS dataset \cite{emiya2010multipitch} (subset MUS) at a sampling rate of 22.05 kHz. 
We choose a constant-Q transformed spectrogram representation with a hop size of $1984$.
The range comprises $120$ frequency bins (24 per octave), starting from a minimal frequency of $65.4$ Hz.
The spectrogram is split into n-grams of 32 frames.
The set of transformations applied to the data during training $\Psi_{\text{pshift, tshift}}$ contains transposition by $[-24, 24]$ frequency bins and time shifts by $[-12, 12]$.

After training, all n-grams of the JKUPDD dataset are projected into the transformation-invariant magnitude space.
Using these representations, a self-similarity matrix is built for each piece using the reciprocal of the cosine distance.
The matrices are then filtered with an identity matrix of size $10\times 10$.
Then, their main diagonals are set to zero. 
Finally, the matrices are first normalized and then centered by subtracting their medians.

For finding repeated sections, the method proposed in \cite{lattner2018learning} is adopted, which finds diagonals in a self-similarity matrix using a threshold.
As we normalized the matrices to zero median, the threshold chosen in this experiment is close to zero (i.e., $0.01$).

\subsubsection{Results and Discussion}
Table \ref{tab:secdiscovery} shows the results of the experiment.
Using our method, we could slightly outperform the Gated Autoencoder approach proposed in \cite{lattner2018learning}.
By visual inspection of the self-similarity matrix, we noted very precise diagonals at repetitions, while almost no similarity is indicated on other parts (this is different from the self-similarity plots provided in \cite{lattner2018learning}).
This selectivity, which may also result from the cosine distance, probably contributes to the slightly higher precision of the proposed method.

\subsection{Invariant Audio-to-Score Alignment}
\begin{table}
\begin{center}
\footnotesize
\begin{tabular}{clcc}
\toprule
\textbf{ID} & \textbf{Dataset} & \textbf{Files} & \textbf{Duration}\\ \midrule
\textbf{CE} & Chopin Etude  & 22 &  $\sim$ 30 min.\\ 
\textbf{CB} & Chopin Ballade & 22 & $\sim$ 48 min. \\ 
\textbf{MS} & Mozart Sonatas & 13 &  $\sim$ 85 min.\\ 
\textbf{RP} & Rachmaninoff Prelude & 3 & $\sim$ 12 min. \\ 
\textbf{B3} & Beethoven 3 & 4 & $\sim$ 52 min. \\ 
\textbf{M4} & Mahler 4 & 4 & $\sim$ 58 min. \\ 
\bottomrule
\end{tabular}
\caption{The evaluation data set for the alignment experiments (see text).}
\vspace{-5mm}
\label{tab:data_set_alignment}
\end{center}
\end{table}
The task of synchronising an audio recording of a music performance and its score has already been studied extensively (see e.g. \cite{HuDT03_audiomatching_WASPAA,EllisP07_CoverSong_ICASSP,JoderER10_MusicAlignment_ICASSP, Mueller15_FMP_SPRINGER,dixon:ismir:2005,grachten:ismir:2013}). Here, we compare synchronisation results using the proposed method ({CAE}) to traditional Chroma features and the {GAE} features introduced in \cite{lattner2018learning}, which were used for music alignment in \cite{ArztL18_AudioToScore_ISMIR}.

For the alignment experiments we follow \cite{ArztL18_AudioToScore_ISMIR}, using the same setup and the same data (see Table \ref{tab:data_set_alignment} for a summary). \emph{CB} and \emph{CE} consist of 22 recordings of excerpts of the Ballade Op. 38 No. 1 and the Etude Op. 10 No. 3 by Chopin \cite{vienna4x22}, \emph{MS} contains performances of the first movements of the piano sonatas KV279-284, KV330-333, KV457, KV475 and KV533 by Mozart \cite{widmer:aimag}, and \emph{RP} consists of three performances of the Prelude Op. 23 No. 5 by Rachmaninoff \cite{arzt:2016}. Finally, \emph{B3} and \emph{M4} are annotated recordings of Beethoven's $3^{rd}$ and Mahler's $4^{th}$ symphonies. Note that CB, CE, MS, and RP consist of piano music, while B3 and M4 consist of orchestral music, but we will use the same model for the whole data set, which was trained on piano music only.

The scores are provided in the MIDI format, with the global tempi set such that the scores roughly match the average length of the given performances, i.e., both representations have the same average tempo, but there still exist substantial differences in local tempi. The scores are then synthesized with the help of \emph{timidity}\footnote{\url{https://sourceforge.net/projects/timidity/}} and a publicly available sound font. The resulting audio files are used as score representations for the alignment experiments. To compute the alignments, a multi-scale variant of the dynamic time warping (DTW) algorithm (see \cite{Mueller15_FMP_SPRINGER} for a detailed description of DTW) is used, namely FastDTW\cite{salvador:ida:2007} with the radius parameter set to 50. 

The CAE is trained the same way and on the same data as described in Section \ref{sec:disc-repe-them} but here we choose a CQT hop size of $448$.
Furthermore, for this experiment, we use an n-gram size of 8. 
The set of transformations applied to the data during training $\Psi_{\text{pshift}}$ are transpositions by $[-24, 24]$ frequency bins.

\subsubsection{Results and Discussion}
\begin{table}
    \centering
    \scriptsize
    \begin{tabular}{llcccc}
    \toprule
    & & \multicolumn{3}{c}{\textbf{`Un-transposed' Data}} & {\textbf{Transp.}}  \\ \cmidrule(r){3-5}  \cmidrule(r){6-6}
        \textbf{DS}                  & \textbf{Metric}                & \textbf{Chroma}                                    & \textbf{GAE}  & \textbf{CAE}  & \textbf{CAE}             \\ \midrule
        \multirow{5}{*}{\textbf{CB}} & 1\textsuperscript{st} Quartile           & \SI{15}{\milli\second}   & \textbf{\SI{10}{\milli\second}} & \textbf{\SI{10}{\milli\second}} & \SI{10}{\milli\second}   \\
                            & Median                 & \SI{34}{\milli\second}   & \SI{22}{\milli\second} & \textbf{\SI{21}{\milli\second}} & \SI{21}{\milli\second} \\
                            & 3\textsuperscript{rd} Quartile           & \SI{80}{\milli\second}  & {\SI{39}{\milli\second}}  & \textbf{\SI{37}{\milli\second}} & \SI{38}{\milli\second}  \\
                            & Err. $\leq$ \SI{50}{\milli\second}  & 64\%                               & 83\%   &  \textbf{84\%}   &  84\%             \\
                            & Err. $\leq$ \SI{250}{\milli\second} & 85\%                            & \textbf{94\%}    & \textbf{94\%}   & 94\%             \\ \midrule
        \multirow{5}{*}{\textbf{CE}} & 1\textsuperscript{st} Quartile           & \SI{13}{\milli\second}   & \textbf{\SI{10}{\milli\second}} & \textbf{\SI{10}{\milli\second}} & {\SI{9}{\milli\second}}\\
                            & Median                 & \SI{29}{\milli\second}  & {\SI{21}{\milli\second}} & \textbf{\SI{19}{\milli\second}} & {\SI{18}{\milli\second}}\\
                            & 3\textsuperscript{rd} Quartile           & \SI{56}{\milli\second}  & {\SI{36}{\milli\second}} & \textbf{\SI{32}{\milli\second}} & {\SI{30}{\milli\second}} \\
                            & Err. $\leq$ \SI{50}{\milli\second}  & 71\%                                & {87\%}   & \textbf{90\%}  & {91\%}               \\
                            & Err. $\leq$ \SI{250}{\milli\second} & 94\%                                  & \textbf{96\%}      & \textbf{96\%} & {97\%}             \\ \midrule
        \multirow{5}{*}{\textbf{MS}} & 1\textsuperscript{st} Quartile           & \SI{7}{\milli\second}   & \textbf{\SI{6}{\milli\second}} & \textbf{\SI{6}{\milli\second}} & {\SI{6}{\milli\second}} \\
                            & Median                 & \SI{16}{\milli\second}  & {\SI{13}{\milli\second}} & \textbf{\SI{12}{\milli\second}} & {\SI{12}{\milli\second}}\\
                            & 3\textsuperscript{rd} Quartile           & \SI{31}{\milli\second}   & {\SI{25}{\milli\second}} & \textbf{\SI{22}{\milli\second}} & {\SI{22}{\milli\second}} \\
                            & Err. $\leq$ \SI{50}{\milli\second} & 85\%                                  & {90\%}     & \textbf{91\%}     & {92\%}            \\
                            & Err. $\leq$ \SI{250}{\milli\second} & 98\%                                & \textbf{100\%}      & \textbf{100\%}    & {99\%}      \\ \midrule
        \multirow{5}{*}{\textbf{RP}} & 1\textsuperscript{st} Quartile           & \SI{17}{\milli\second}   & {\SI{14}{\milli\second}} & \textbf{\SI{9}{\milli\second}} & {\SI{9}{\milli\second}} \\
                            & Median                 & \SI{43}{\milli\second}   & {\SI{34}{\milli\second}} & \textbf{\SI{20}{\milli\second}} & {\SI{21}{\milli\second}}  \\
                            & 3\textsuperscript{rd} Quartile           & \SI{113}{\milli\second} & {\SI{90}{\milli\second}} & \textbf{\SI{55}{\milli\second}} & {\SI{69}{\milli\second}} \\
                            & Err. $\leq$ \SI{50}{\milli\second}  & 55\%                                 & {63\%}       & \textbf{74\%}    & {70\%}       \\
                            & Err. $\leq$ \SI{250}{\milli\second} & {91\%}                                & 90\%    & \textbf{95\%}    & 93\%            \\  \midrule
        \multirow{5}{*}{\textbf{B3}} & 1\textsuperscript{st} Quartile           & {\SI{20}{\milli\second}}   & \SI{25}{\milli\second} & \textbf{\SI{17}{\milli\second}} & {\SI{18}{\milli\second}} \\
                            & Median                 & {\SI{48}{\milli\second}}  &\SI{54}{\milli\second}  & \textbf{\SI{39}{\milli\second}} & {\SI{42}{\milli\second}} \\
                            & 3\textsuperscript{rd} Quartile           & {\SI{108}{\milli\second}} &{\SI{104}{\milli\second}} & \textbf{\SI{83}{\milli\second}} & {\SI{99}{\milli\second}} \\
                            & Err. $\leq$ \SI{50}{\milli\second}   & {52\%}                &    47\%         & \textbf{59\%}   &    56\%                 \\
                            & Err. $\leq$ \SI{250}{\milli\second} & {88\%}             &      {90\%}            & \textbf{91\%}     &    88\%               \\ \midrule
        \multirow{5}{*}{\textbf{M4}} & 1\textsuperscript{st} Quartile           & {\SI{46}{\milli\second}}  &  \SI{50}{\milli\second}  & \textbf{\SI{42}{\milli\second}}  & {\SI{46}{\milli\second}} \\
                            & Median                 & {\SI{110}{\milli\second}} &  \SI{129}{\milli\second}  & \textbf{\SI{99}{\milli\second}} & {\SI{110}{\milli\second}} \\
                            & 3\textsuperscript{rd} Quartile           & {\SI{278}{\milli\second}} &  \SI{477}{\milli\second}  & \textbf{\SI{255}{\milli\second}} & {\SI{290}{\milli\second}} \\
                            & Err. $\leq$ \SI{50}{\milli\second}   & {27\%}                     &     25\%       & \textbf{29\%}   & {27\%}               \\
                            & Err. $\leq$ \SI{250}{\milli\second} & {73\%}                  &      66\%      & \textbf{75\%}     & {72\%}             \\ 
                            \bottomrule
    \end{tabular}
    \caption{Comparison of the proposed features {CAE} to {Chroma} features and features computed via a gated autoencoder {GAE}. The first three columns show results on normal, i.e., un-transposed data. The rightmost column shows the average result of alignments of the original performances to scores in 12 different transpositions.}
        \label{tab:alignment_1}
\vspace{-2mm}
\end{table}

In the alignment experiments, we compare the proposed {CAE} features to the results presented in \cite{ArztL18_AudioToScore_ISMIR}, where {Chroma} features and features computed via a gated autoencoder (GAE) were compared to each other. Table \ref{tab:alignment_1} gives an overview of the results. The first three columns show that the proposed {CAE} features consistently outperform the other two methods in the normal alignment setting (i.e., without any transpositions). Additionally, the rightmost column shows that for {CAE}, the results essentially stay the same, even when the alignment is computed with transposed versions of the score. This demonstrates the invariance to transpositions, which is a serious advantage over the Chroma features.

As has been shown in \cite{ArztL18_AudioToScore_ISMIR}, the GAE features are highly sensitive to tempo differences between the score representation and the performance. To see if the proposed {CAE} features suffer from the same problem, we repeated this experiment and performed alignments on artificially slowed-down and sped-up score representations. The results are shown in Table \ref{tab:alignment_2}. For all tested features, the degree to which the tempi of the score representation and the performance match influences the alignment quality. The experiments suggest that {CAE} is less sensitive to differences in tempi than {GAE}, but the Chroma features still have the advantage over GAE in this matter. We also conducted experiments with more extreme tempi, which further confirmed this trend. The reason for the higher robustness to tempo differences of the CAE features over the GAE features may be found in the way the GAE features are computed. In a GAE, two inputs $\{\mathbf{x}_{t-n, \dots, t}, \mathbf{x}_{t+1}\}$ are compared to one another, and the features are sensitive to the position and order of events in $\mathbf{x}_{t-n, \dots, t}$. When training a CAE only for transposition-invariance, the resulting features represent mainly distances in the frequency-dimension of the input and tend to be invariant to the position of events in time.

\begin{table*}[]
    \centering
    \scriptsize
    \begin{tabular}{llccccccccc}
    \toprule
    & &  \multicolumn{3}{c}{\textbf{Chroma}} & \multicolumn{3}{c}{\textbf{GAE}} & \multicolumn{3}{c}{\textbf{CAE}}  \\ \cmidrule(r){3-5} \cmidrule(r){6-8} \cmidrule(r){9-11}
        \textbf{DS}                  & \textbf{Metric}                & \textbf{$\mathbf{\frac{2}{3}}$  T.}                & \textbf{Base T.}                    &\textbf{$\mathbf{\frac{4}{3}}$  T.}       & \textbf{$\mathbf{\frac{2}{3}}$  T.}                   & \textbf{Base T.}                    &\textbf{$\mathbf{\frac{4}{3}}$  T.} &  \textbf{$\mathbf{\frac{2}{3}}$  T.}                   & \textbf{Base T.}                    &\textbf{$\mathbf{\frac{4}{3}}$  T.}             \\ \midrule
        \multirow{2}{*}{\textbf{CB}}
                            & Error $\leq$ \SI{50}{\milli\second}  & 54\%  & 64\%  & 67\%  & 47\%  & 83\%  & 33\%  & \textbf{80\%}  & \textbf{84\%}  & \textbf{85\%}                           \\
                            & Error $\leq$ \SI{250}{\milli\second}   & 82\%  & 85\%  & 85\%  & 87\%  & \textbf{94\%}  & 84\% & \textbf{91\%} & \textbf{94\%} & \textbf{94\%}                      \\ \midrule
        \multirow{2}{*}{\textbf{CE}}
                            & Error $\leq$ \SI{50}{\milli\second} & 69\%  & 71\%  & 73\%  & 40\%  & 87\%  & 38\% & \textbf{85\%} & \textbf{90\%} & \textbf{88\%}                               \\
                            & Error $\leq$ \SI{250}{\milli\second} & 90\%  & 94\%  & 94\%  & \textbf{93\%}  & \textbf{96\%}  & 80\%  & \textbf{93\%} & \textbf{96\%} & \textbf{95\%}                        \\ \midrule
        \multirow{2}{*}{\textbf{MS}}
                            & Error $\leq$ \SI{50}{\milli\second} & 79\%  & 85\%  & 75\%  & 84\%  & 90\%  & 74\% & \textbf{86\%}  & \textbf{91\%}  & \textbf{76\%}                              \\
                            & Error $\leq$ \SI{250}{\milli\second} & 98\%  & 98\%  & 97\%  & \textbf{99\%}  & \textbf{100\%}  & \textbf{98\%} & \textbf{99\%}  & \textbf{100\%}  & \textbf{98\%}                        \\ \midrule
        \multirow{2}{*}{\textbf{RP}} 
                            & Error $\leq$ \SI{50}{\milli\second}  & 53\%  & 55\%  & 56\%  & 43\%  & 63\%  & 37\%  & \textbf{67\%}  & \textbf{74\%}  & \textbf{63\%}                          \\
                            & Error $\leq$ \SI{250}{\milli\second} & 92\%  & 91\%  & 87\%  & 82\%  & 90\%  & 85\%  & \textbf{95\%}  & \textbf{95\%}  & \textbf{91\%}                        \\  \midrule
        \multirow{2}{*}{\textbf{B3}} 
                            & Error $\leq$ \SI{50}{\milli\second}  &\textbf{44\%}  & {52\%}               & \textbf{36\%} & -- &    47\%      & -- &   {39\%}  & \textbf{59\%}  &     {33\%}             \\
                            & Error $\leq$ \SI{250}{\milli\second}& \textbf{83\%}  & {88\%}           & \textbf{82\%}& -- &      {90\%}       &--&   {82\%}   & \textbf{91\%}    &  \textbf{82\%}           \\ \midrule
        \multirow{2}{*}{\textbf{M4}} 
                            & Error $\leq$ \SI{50}{\milli\second}   &\textbf{26\%}& {27\%}                   &\textbf{24\%}&-- &     25\%      &--&  {24\%} & \textbf{29\%}   &  {22\%}              \\
                            & Error $\leq$ \SI{250}{\milli\second}& \textbf{75\%}& {73\%}                  & \textbf{71\%}&--&      66\%      &--&  {72\%}& \textbf{75\%}     &    {65\%}         \\ 
                            \bottomrule
    \end{tabular}
    \caption{Results on score representations with different tempi (higher is better). \emph{Base T.} refers to a globally set tempo that ensures that the duration of the score representation is roughly equal to the duration of a typical performance. \emph{$\mathbf{\frac{2}{3}}$  T.}  and \emph{$\mathbf{\frac{4}{3}}$  T.} refer to score representation with the tempo set to $\mathbf{\frac{2}{3}}$ and $\mathbf{\frac{4}{3}}$ of the base tempo. The two metrics used are the percentage of events that are aligned with an error lower or equal \SI{50}{\milli\second} and \SI{250}{\milli\second} (i.e. higher is better). The missing numbers for {GAE} were not provided in \cite{ArztL18_AudioToScore_ISMIR}.}
        \label{tab:alignment_2}
\vspace{-4mm}
\end{table*}

\subsection{Classification of MNIST digits}
We test the ability of the CAE to learn rotation-invariance in 2D images using randomly rotated MNIST digits (the dataset was first described in \cite{LarochelleECBB07}).
Given the set of rotations $\Psi_\text{rot}$ with rotation angles $[0, 2\pi)$ about the origin of the images.
For any MNIST instance $\mathbf{x}_k$, we create a rotated version $\psi_i(\mathbf{x}_k)$ and a further rotated version $\psi_j(\psi_i(\mathbf{x}_k))$, where $\psi_i, \psi_j \in \Psi_\text{rot}$, resulting in pairs $\{\psi_i(\mathbf{x}_k), \psi_j(\psi_i(\mathbf{x}_k))\}$.
After the CAE is trained on $50$k such pairs, single randomly rotated instances are projected into the magnitude space.
On these projections, a logistic regression classifier is trained to predict the class labels.
We test different train set sizes (sampled from the main train set with balanced class distribution).
50-fold cross-validation is used, where evaluation is always performed on 10k test instances, independent of the train set size.
For comparison, we perform k-nn classification on the randomly rotated images (i.e., the input space), and unrotated images directly.
We choose logistic regression for the magnitude space and k-nn classification for the input space because they showed the overall best results for those representations.
This choice, as well as the overall experiment setup,  reflects that in \cite{memisevic2013aperture}.

\subsubsection{Results and Discussion}
\begin{figure}
\centering
\includegraphics[width=0.8\linewidth]{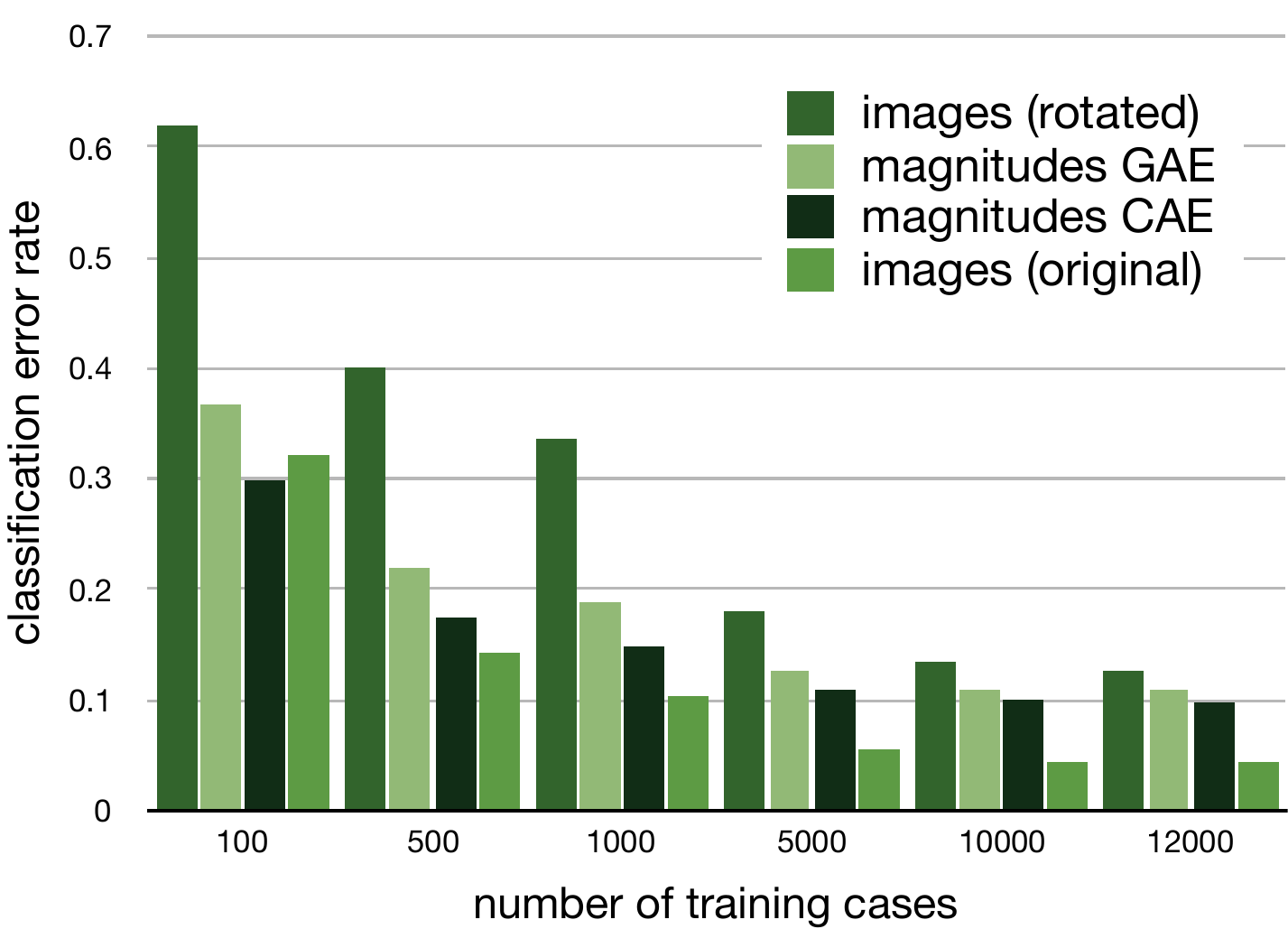} \\
\caption{Classification error rates in the input space (images) and the magnitude space (magnitudes) on the rotated MNIST dataset with different train set sizes. ``Images (original)'' denotes the results of the unrotated MNIST dataset for comparison.}
\label{fig:mnist_rot_chart}
\vspace{-.5mm}
\end{figure}

\begin{figure}
\begin{tabular}{cc}
\hspace{-4mm}
\includegraphics[trim=30mm 23mm 38mm 40mm, clip, width=.50\linewidth]{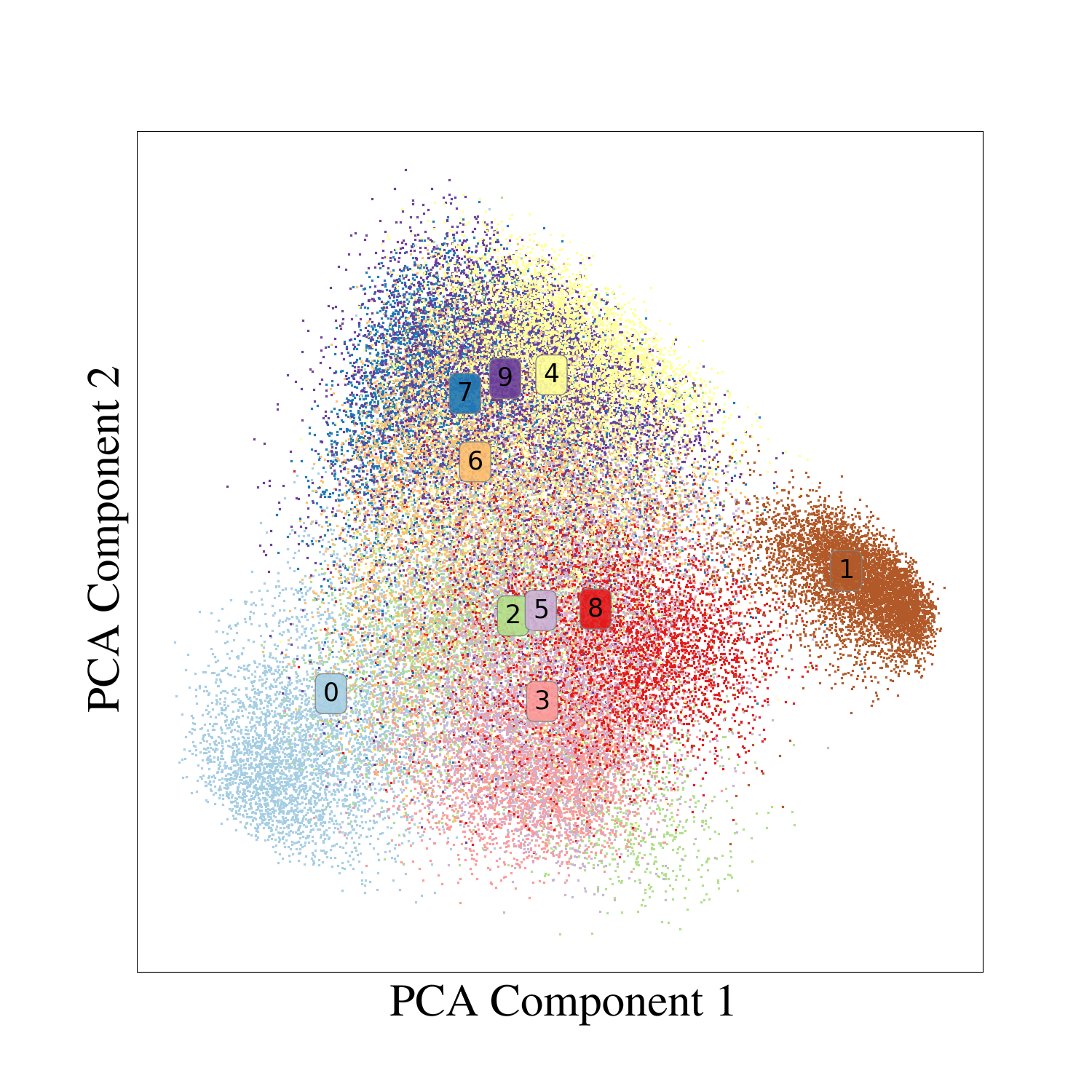} & \hspace{-4mm}
\includegraphics[trim=30mm 23mm 38mm 40mm, clip, width=.50\linewidth]{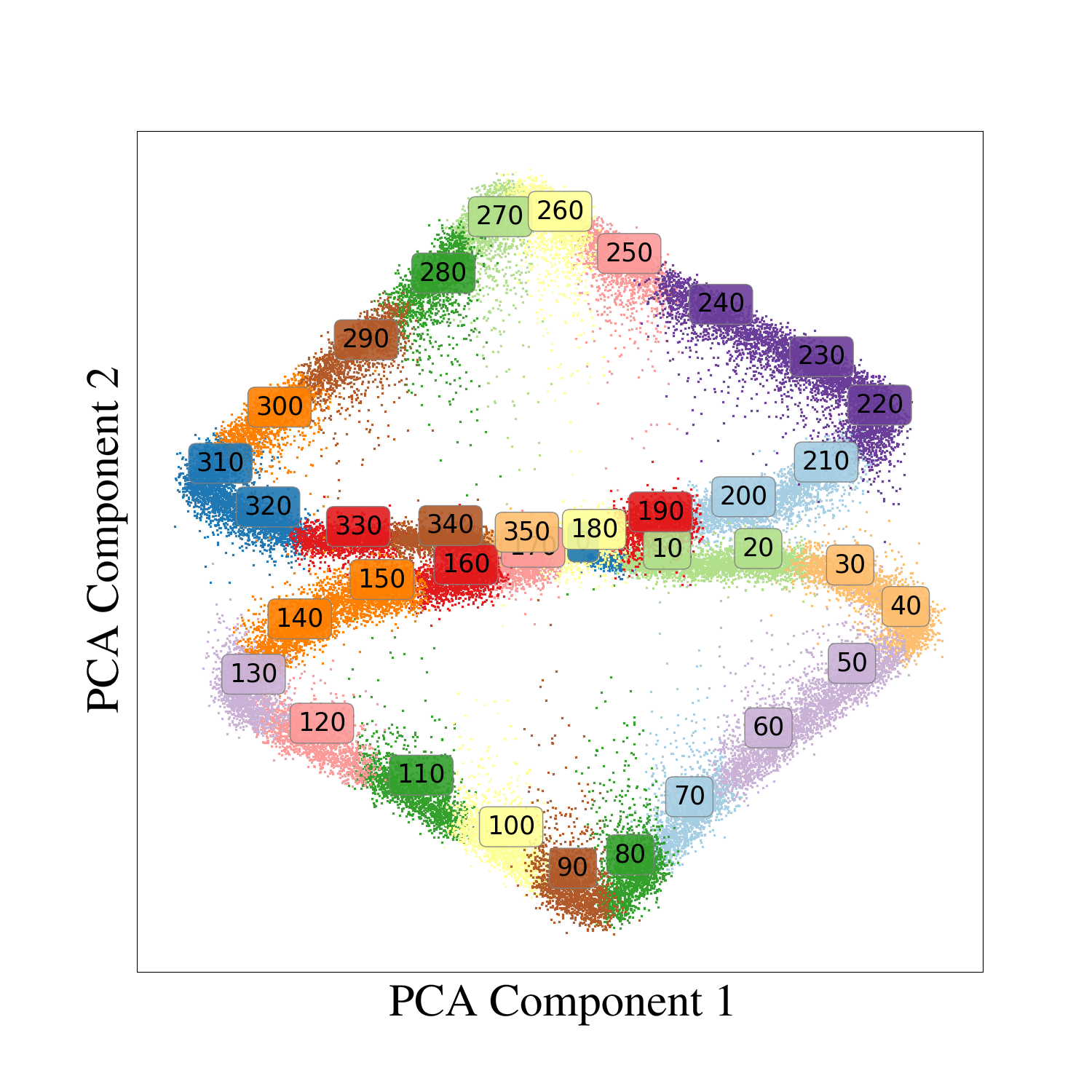} \\
(a) & (b) \\
\end{tabular}
\vspace{-3mm}
\caption{PCAs of rotated MNIST digits in the magnitude space (a) and the phase-difference space (b) (best viewed in color). The magnitude space represents the data in the absence of the transformations leading to clusters of the digit classes (colored and labeled accordingly). The phase-difference space represents the transformations between images, independent of their identity (colors and labels denote rotation angles quantized into 36 bins).}
\label{fig:pca}
\vspace{-.3cm}
\end{figure}

Figure \ref{fig:mnist_rot_chart} shows the results of the rotated MNIST classification task.
The error rates of \emph{magnitudes GAE} were obtained by a GAE architecture which was extended to learn basis functions, as reported in \cite{memisevic2013aperture}.
Classification on the magnitude space of the CAE (\emph{magnitudes CAE}) leads to substantially better results than those of the GAE, even though only $256$ basis elements are used in the CA, compared to $1000$ in the GAE.
This is probably due to the explicit training of the CA to learn an invariant magnitude space, while the magnitude space of the GAE is learned indirectly during the learning of transformations.
Overall, classification on the rotation-invariant magnitude spaces performs much better than on the input space of rotated images in particular for small train set sizes.
The difference in performance between \emph{images (original)} and \emph{magnitudes CA} reflects the gap between a theoretically optimal rotation-invariant representation (as \emph{images (original)} are not rotated), and the representations learned by the CA.
On 100 training cases, logistic regression would outperform the k-nn classification on the input spaces, while for all other train set sizes k-nn is superior over logistic regression.
Thus, we obtain slightly worse classification performance of \emph{images (original)} on 100 training cases compared to \emph{magnitudes CA}.

Figure \ref{fig:pca} shows PCAs of the randomly rotated MNIST digits projected into the magnitude space and the phase-difference space of the CAE.
The clusters in the magnitude space indicate that images with the same content (i.e., class label) yield similar projections, independent of their rotations.
The clusters in the phase-difference space show that phase differences clearly represent the transformations in the data.

\vspace{-3mm}

\section{Conclusion and future work}\label{sec:conclusion}
The empirical results in this work show that for music alignment, structure analysis, and invariant classification tasks, the features learned by the CAE have advantages over other features, like Chroma features, and features learned by a GAE.
As opposed to Chroma features, the CAE features are transposition-invariant, and generally perform better in the alignment task.
Compared to the features learned by a GAE, the CAE features are more robust to differences in tempo between alignment data.

Future work should involve investigating the use of the ``phase-difference'' space of the CAE.
For example, qualifying transformations between sections in music could lead to a richer musical structure analysis (e.g., determining mutually transposed parts, or finding sections with similar rhythm but different tonality).

The learned bases could also be used in scattering transforms (i.e., as convolutional filters).
As opposed to conventional scattering transforms, where the bases are fixed in general, learned bases may help reducing model sizes or to cover different invariances.
Using rotation-invariant filters in convolutional settings may lead to rotation-invariant architectures, similar to what is proposed in \cite{DBLP:conf/cvpr/WorrallGTB17}.

\section{Acknowledgments}
This project has received funding from the European Union's Horizon 2020 research and innovation programme under the Marie Skłodowsa-Curie grant agreement No. 765068.
Monika D{\"o}rfler is supported by the Vienna Science and Technology Fund (WWTF) project SALSA (MA14-018).





\small
\bibliography{bib_sl,bib_aa}

\end{document}